\documentclass[twocolumn,floatfix,superscriptaddress,nofootinbib]{revtex4-1}
\usepackage{amssymb,amsmath,amsthm,amsfonts}
\usepackage{graphicx,xcolor}
\usepackage{hyperref}

 1

\newcommand{\ket}[1]{\left| {#1} \right\rangle}
\newcommand{\bra}[1]{\left\langle {#1}\right|}
\newcommand{\ketbra}[2]{\ket{#1}\!\bra{#2}}
\newcommand{\braket}[2]{\langle #1 | #2 \rangle}
\renewcommand{\t}[1]{\textrm{#1}}
\newcommand{\mat}[2]{\left( \begin{array}{#1} #2  \end{array}\right)}
\DeclareMathOperator{\id}{id}
\DeclareMathOperator{\distance}{d}
\renewcommand{\d}[1]{\distance\!\left(#1\right)}
\newcommand{\dx}[1]{\distance\!\left[#1\right]}
\renewcommand{\r}{r}
\newcommand{\con}{\mathcal{C}^\theta_{\mathrm{con}}}
\newcommand{\conx}{\mathcal{C}_{\mathrm{con}}}
\newcommand{\intd}[2]{\int\!d{#1}\ {#2}}

\makeatletter
\newcommand\gobblepars{\@ifnextchar\par{\expandafter\gobblepars\@gobble}{}}
\makeatother
\newcommand{\sectionprl}[1]{{\it #1.---}\gobblepars}

\newtheorem{lemma}{Lemma}
\theoremstyle{lemma}
\newtheorem{theorem}{Theorem}
\theoremstyle{theorem}

\begin{document}

\title{Using Quantum Metrological Bounds in Quantum Error Correction:\\
A Simple Proof of the Approximate Eastin-Knill Theorem}
\author{Aleksander Kubica}
\affiliation{Perimeter Institute for Theoretical Physics, Waterloo, ON N2L 2Y5, Canada}
\affiliation{Institute for Quantum Computing, University of Waterloo, Waterloo, ON N2L 3G1, Canada}
\author{Rafa{\l} Demkowicz-Dobrza{\'n}ski}
\affiliation{Faculty of Physics, University of Warsaw, Pasteura 5, PL-02-093 Warszawa, Poland}

\begin{abstract}
We present a simple proof of the approximate Eastin-Knill theorem, which connects the quality of a quantum error-correcting code (QECC) with its ability to achieve a universal set of transversal logical gates.
Our derivation employs powerful bounds on the quantum Fisher information in generic quantum metrological protocols to
characterize the QECC performance measured in terms of the worst-case entanglement fidelity.
The theorem is applicable to a large class of decoherence models, including independent erasure and depolarizing noise.
Our approach is unorthodox, as instead of following the established path of utilizing QECCs to mitigate noise in quantum metrological protocols, we apply methods of quantum metrology to explore the limitations of QECCs.
\end{abstract}

\maketitle

Quantum error-correcting codes (QECCs) are an indispensable tool for realizing fault-tolerant universal quantum computation~\cite{Shor1995,Shor1996,Gottesman1996,Campbell2017}.
Using QECCs we can protect encoded information from detrimental effects of decoherence due to unwanted interactions with the environment, as long as the noise is sufficiently weak~\cite{Aharonov1997,Aliferis2005,Aliferis2007}.
Although we want to isolate encoded information from the noisy environment as much as possible, we still wish to be able to easily perform logical operations on it.

One particularly simple way to realize fault-tolerant quantum computation is with the help of QECCs with transversal logical gates~\cite{Bombin2006,Bombin2007,Bombin2013,Kubica2015a,Watson2015,Kubica2015,Vasmer2019}.
Transversal gates, which are tensor products of unitaries acting independently on different subsystems, do not spread errors in an uncontrollable way.
However, there is a no-go theorem, the Eastin-Knill theorem~\cite{Eastin2009,Zeng2011}, which rules out the existence of finite-dimensional quantum error-detecting codes with a universal set of transversal logical gates.

Recently, there has been a lot of interest in quantifying the Eastin-Knill theorem~\cite{Bravyi2013,Pastawski2014,Jochym-OConnor2018}, as well as circumventing it~\cite{Hayden2017,Faist2019,Woods2020,Wang2019}.
In particular, Refs.~\cite{Faist2019,Woods2020} introduce a version of the approximate Eastin-Knill theorem, which connects code's quality measured in terms of worst-case entanglement fidelity with code's ability to achieve a universal set of transversal logical gates.
However, the aforementioned works rely heavily on the technical tools from representation theory and the notion of reference frames, making for intricate proofs and limiting results to elementary noise models.

In our work, we present an alternative and more streamlined way of proving the approximate Eastin-Knill theorem.
Our derivation utilizes powerful bounds on the optimal performance of quantum metrological protocols in the presence of decoherence \cite{Fujiwara2008, Escher2011,Demkowicz2012, Kolodynski2013,Demkowicz2014, Demkowicz2017, Zhou2018, Zhou2020}.
These metrological bounds are easily computable for a large class of noise models.
Thus, our approach, unlike the previous results, is not limited to the erasure noise and can be straightforwardly applied to, for instance, the depolarizing noise.
To the best of our knowledge, our approach is the first example of quantum metrology helping quantum error correction, as until now only error correction techniques were employed in quantum metrological protocols to mitigate noise and achieve better scaling~\cite{Kessler2014a, Dur2014, Arrad2014, Unden2016, Sekatski2017, Zhou2018, Gorecki2019, Datta2019, Layden2019,  Zhou2020a, Zhou2020}.

In what follows, we first discuss QECCs with transversal logical gates.
Then, we briefly review bounds on the quantum Fisher information (QFI) accessible in an arbitrary quantum metrological protocol in the presence of local noise.
Next, we prove (technical) Lemma~\ref{lemma}, which provides a lower bound on the QFI in any
parameter-dependent channel in terms of the Bures distance of that channel from the ideal unitary rotation.
Lemma~\ref{lemma} can be regarded as a result on its own, and we expect it to be useful also beyond the scope of our work.
Lastly, we apply the metrological bounds to limit the quality of the QECCs and arrive at the main result, (approximate Eastin-Knill) Theorem~\ref{theorem}.

\sectionprl{Quantum error-correcting codes}

Consider a QECC, where states and operations on a logical system $L$ are encoded into a physical system $A$, comprising $n$ disjoint finite-dimensional subsystems $A_1,\dots,A_n$, via an encoding channel $\mathcal{E}_{L \rightarrow A}$.
The system $A$ is subject to noise represented by a channel $\mathcal{N}_{A}$.
We consider a local noise model acting on each subsystem independently, i.e.,
\begin{equation}
\mathcal{N}_{A} = \bigotimes_{i=1}^n \mathcal{N}_{A_i}.
\end{equation}
This is a natural assumption
in quantum error correction, and will also allow us to use powerful quantum metrological bounds in a straightforward way \cite{Fujiwara2008, Escher2011,Demkowicz2012, Kolodynski2013,Demkowicz2014, Demkowicz2017, Zhou2018, Zhou2020}.
Let $\mathcal{R}_{A\rightarrow L}$ be a recovery map, which detects and attempts to correct errors
in a way that the composed channel
\begin{equation}
\mathcal{I}_L = \mathcal{R}_{A\rightarrow L} \circ \mathcal{N}_{A} \circ \mathcal{E}_{L\rightarrow A}
 \end{equation}
is as close as possible to the identity channel $\id_L$ on the logical system $L$.
More precisely, we want to maximize the worst-case entanglement fidelity $\mathcal{F}(\mathcal{I}_L,\id_L)$, defined for any two channels $\mathcal{C}$ and $\mathcal{D}$ as follows~\cite{Schumacher1996, Gilchrist2005}
\begin{equation}
\mathcal{F}(\mathcal{C},\mathcal{D})
= \min_{\ket{\Phi}} f[\mathcal{C}\otimes\id_{L'}(\ketbra{\Phi}{\Phi}), \mathcal{D}\otimes\id_{L'}(\ketbra{\Phi}{\Phi})],
\end{equation}
where $\ket{\Phi}$ is an arbitrary (in principle entangled) pure state supported on $L$ and a reference system $L'$, and $f(\rho,\sigma) = \mathrm{Tr}(\sqrt{\!\sqrt{\rho}\sigma\sqrt{\rho}})$ denotes the fidelity~\cite{Uhlmann1976} between two states $\rho$ and $\sigma$.
We use $\mathcal{F}(\mathcal{C},\mathcal{D})$ to define the Bures distance~\cite{Bures1969} between $\mathcal{C}$ and $\mathcal{D}$, namely
\begin{equation}
\d{\mathcal{C},\mathcal{D}} = \sqrt{1 - \mathcal{F}(\mathcal{C},\mathcal{D})}.
\end{equation}
Finally, we say that the code $\mathcal{E}_{L\rightarrow A}$ is $\epsilon$-correctable~\cite{Beny2010} under the noise $\mathcal{N}_{A}$ if
there exists a recovery operation $\mathcal{R}_{A\rightarrow L}$ such that
\begin{equation}
\label{eq:error}
\d{\mathcal{I}_L,\id_L} \leq \epsilon.
\end{equation}

\sectionprl{Formulation of the problem}

We are interested in QECCs with transversal logical gates forming a universal gate set.
By definition, a transversal logical unitary $U_L$ is represented as a tensor product of unitaries on different physical subsystems, i.e., $U_A = \bigotimes_{i=1}^n U_{A_i}$.
First, instead of considering the full universal set of gates, we will restrict our attention to
 a family of logical operators $\{U_L^\theta\}_\theta$ parameterized by $\theta\in[0,2\pi]$ and generated by a Hermitian generator $T_L$ via
\begin{equation}
\label{eq_rotation_logical}
U_L^\theta = \exp(-i \theta T_L).
\end{equation}
Note that $\{U_L^\theta\}_\theta$ may be regarded as a representation of the group $U(1)$.
 In what follows we will focus our attention on $U(1)$-covariant codes.\footnote{
We say that a code $\mathcal{E}_{L\rightarrow A}$ is $G$-covariant iff for any $g\in G$
\begin{equation*}
\mathcal{E}_{L \rightarrow A}(U_L(g)\cdot U_L(g)^\dagger) = U_A(g)
 \mathcal{E}_{L \rightarrow A}(\cdot) U_A(g)^\dagger,
\end{equation*}
where $G$ is a Lie group with representations $U_L(g)$ and $U_A(g)$ acting unitarily on the logical and physical systems.
The notion of $U(d_L)$-covariance, where $d_L$ is the dimension of the logical Hilbert space, is related to a universal transversal gate set; see Appendix E in~\cite{Faist2019}.
Note that $U(d_L)$-covariance implies $U(1)$-covariance.}
Combining covariance and transversality we get
\begin{equation}
\label{eq_covariance}
\mathcal{E}_{L \rightarrow A}\circ \mathcal{U}_L^\theta = \mathcal{U}_A^\theta\circ  \mathcal{E}_{L \rightarrow A},
\end{equation}
where $\mathcal{U}_L^\theta$ and $\mathcal{U}_A^\theta$ denote channels on the systems $L$ and $A$ corresponding to the unitary rotations $U_L^\theta$ and $U_A^\theta$, such that $U_A^\theta$ can be expressed in terms of local Hermitian generators $T_{A_i}$ acting on physical subsystems $A_i$, namely
\begin{equation}
U_A^\theta = \bigotimes_{i=1}^n  U_{A_i}^\theta = \exp\left( -i \theta \sum_{i=1}^n T_{A_i}\right).
\end{equation}

The physical system $A$ is now effectively subject to a channel $\mathcal{N}^\theta_{A}$ defined as
\begin{equation}
\label{eq_noise}
\mathcal{N}_{A}^\theta = \bigotimes_{i=1}^n (\mathcal{N}_{A_i}\circ\mathcal{U}_{A_i}^\theta),
\end{equation}
where $\mathcal{N}_{A_i}$ represents the noise process independent of the parameter $\theta$ and acting on the physical subsystem $A_i$.
Since $\mathcal{N}^\theta_{A}$ is a composition of a $\theta$-dependent unitary rotation channel with local noise, thus it can be viewed as a noisy parameter encoding channel.
Fundamental quantum metrological bounds on the QFI impose restrictions on how much we can learn about $\theta$ by performing any strategy, in particular the optimal one, and this will lead us to limitations on the quality of the QECC.

\sectionprl{Quantum metrological bounds}

A key concept in quantum metrology \cite{Giovannetti2006, Paris2009, Giovannetti2011, Toth2014, Demkowicz2015, Pezze2018, Pirandola2018}
as well as for our discussion of quantum metrological bounds is the quantum Fisher information (QFI).
For a given family of states $\rho^\theta$, the inverse of the QFI $F(\rho^\theta)$ provides a lower bound on the
variance of a parameter $\theta$ for an arbitrary measurement and (locally unbiased) estimation procedure performed on $\rho^\theta$~\cite{Helstrom1976, Holevo1982, Braunstein1994, Hayashi2005}, namely
 \begin{equation}
 \Delta^2 \theta  \geq \frac{1}{F(\rho^\theta)}, \quad F(\rho^\theta)= \t{Tr}\left[\rho^\theta (L^\theta)^2  \right].
 \end{equation}
The definition of the QFI above involves the symmetric logarithmic derivative operator $L^\theta$ implicitly defined via
$\dot{\rho}^\theta= \frac{1}{2}\left(\rho^\theta L^\theta + L^\theta \rho^\theta \right)$, where $\dot{\rho}^\theta$ denotes the derivative of $\rho^\theta$ with respect to $\theta$.
In particular, the QFI for $\rho^\theta = \ketbra{\psi^\theta}{\psi^\theta}$ is given by
\begin{equation}
F(\ket{\psi^\theta}) = 4 \left(\braket{\dot{\psi}^\theta}{\dot{\psi}^\theta}  - \left|\braket{\dot{\psi}^\theta}{\psi^\theta}\right|^2 \right).
\end{equation}
Thus, if the parameter $\theta$ is imprinted unitarily on the pure state $\ket{\psi^\theta} = e^{-i T \theta} \ket{\psi}$, then the QFI is proportional to the variance of the generator $T$, i.e.,
\begin{equation}
F(\ket{\psi^\theta}) = 4( \bra{\psi} T^2 \ket{\psi} - \bra{\psi} T \ket{\psi}^2).
\end{equation}
Intuitively, the QFI quantifies `how fast' a quantum state $\ket{\psi^\theta}$ changes with the change of the parameter $\theta$.
The faster $\ket{\psi^\theta}$ changes the easier it is to estimate $\theta$, since the states $\ket{\psi^\theta}$  and $\ket{\psi^{\theta+d\theta}}$ become more distinguishable for a given small change $d\theta$ of $\theta$.

While for noiseless metrological models one can easily identify fundamental limits as well as the optimal
 protocols that maximize the respective QFI \cite{Bollinger1996, Giovannetti2006}, the same task in the case
 of noisy models is highly nontrivial \cite{Huelga1997, Shaji2007, Dorner2008}.
 This prompted development of efficient methods to compute fundamental metrological
 bounds in the presence of noise \cite{Fujiwara2008, Escher2011,Demkowicz2012, Kolodynski2013,Demkowicz2014, Demkowicz2017, Zhou2018, Zhou2020}.

In order to apply the bounds to our scenario, we first express each channel $\mathcal{N}_{A_i}\circ\mathcal{U}_{A_i}^\theta$ acting on the physical subsystem $A_i$ in terms of its Kraus operators as
\begin{equation}
\mathcal{N}_{A_i}\circ\mathcal{U}_{A_i}^\theta(\cdot) = \sum_{k} K_{i,k}^\theta \cdot K_{i,k}^{\theta\dagger}, \quad K_{i,k}^\theta = K_{i,k} e^{-i T_{A_i} \theta},
\end{equation}
where we explicitly separate the $\theta$-dependent part $e^{-i T_{A_i} \theta}$ and the $\theta$-independent part $K_{i,k}$ of the Kraus operator $K_{i,k}^\theta$.
The bounds are formulated in terms of a minimization over different equivalent Kraus representations $\{K_{i,k}^\theta\}_k$ of $\mathcal{N}_{A_i}\circ\mathcal{U}_{A_i}^\theta$.
Since the noisy channel $\mathcal{N}^\theta_{A}$ acts independently on each of the physical subsystems,
the bounds can be calculated efficiently and depend only on the number of applications of a given single subsystem noisy channel $\mathcal{N}_{A_i}\circ\mathcal{U}_{A_i}^\theta$, irrespectively of the other features of the protocol, such as the initial entanglement between the physical subsystems or additional intermediate measurements and adaptive procedures~\cite{Demkowicz2014}.

\begin{figure}[t!]
\centering
\includegraphics[width=\columnwidth]{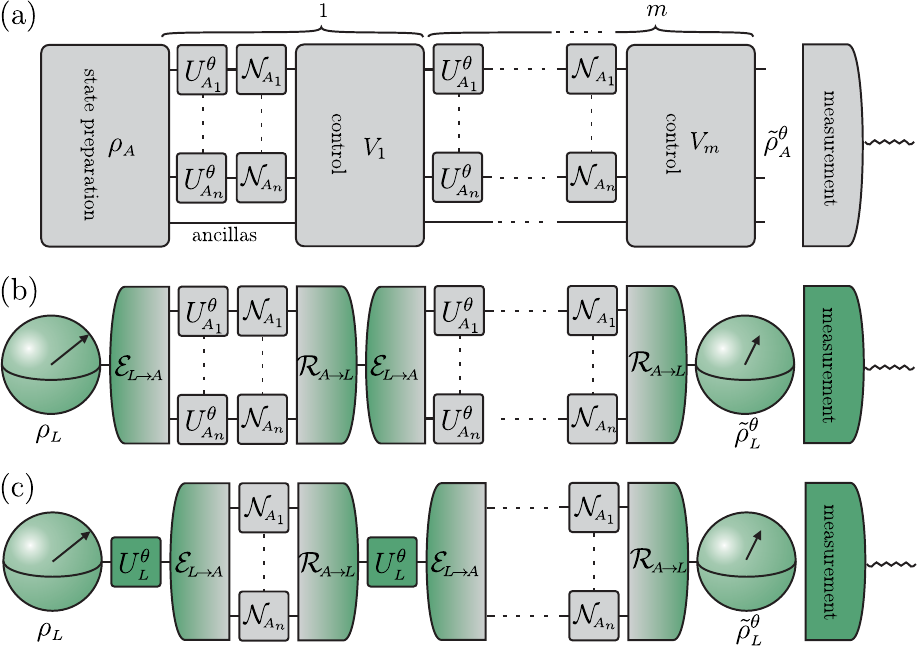}
\caption{
(a) A generic quantum metrological protocol involving $m$ sequential applications of the parameter encoding channel $\bigotimes_{i=1}^n (\mathcal{N}_{A_i}\circ \mathcal{U}_{A_i}^\theta)$ interleaved with the control operations.
(b) A metrological protocol with the initial state $\rho_A = \mathcal{E}_{A\rightarrow L}(\rho_L)$ encoded into a QECC and the control corresponding to a composition of the recovery $\mathcal{R}_{A\rightarrow L}$ and encoding $\mathcal{E}_{A\rightarrow L}$ maps.
(c) Using the covariance property of $\mathcal{E}_{L\rightarrow A}$ we find an equivalent scheme where the parameter-dependent unitary $\bigotimes_{i=1}^n U^\theta_{A_i}$ on the physical system $A$ (shaded in gray) is expressed as $U^\theta_L$ on the logical system $L$ (shaded in green).}
\label{fig_schemes}
\end{figure}
Let us now consider a metrological scheme depicted in Fig.~\ref{fig_schemes}(a), which involves preparing an arbitrary state $\rho_A$, applying $m$ times the noisy channel $\mathcal{N}^\theta_{A}$ interleaved with the (adaptive) control $\{V_i\}_i$, and performing an arbitrary measurement of the output state $\tilde{\rho}^\theta_A$.
The bound on the maximal QFI of the output state takes the form
\begin{equation}
\label{eq_m_uses}
\max_{\rho_A, \{V_i\}_i } F(\tilde{\rho}^\theta_A) \leq m F^{\uparrow},\quad
F^{\uparrow} = 4\sum_{i=1}^n \min_{\{K_{i,k}^\theta\}_k, \beta_i = 0} \| \alpha_i \|,
\end{equation}
where $\alpha_i = \sum_k \dot{K}^{\theta\dagger}_{i,k} \dot{K}^\theta_{i,k}$, $\beta_i = \sum_{k} \dot{K}_{i,k}^\dag K_{i,k}$,
 $\dot{K}^\theta_{i,k}$ is the derivative of $K^\theta_{i,k}$ with respect to $\theta$ and $\| \cdot \|$ denotes the operator norm; for completeness, we include a derivation of the bound in Appendix~\ref{sec:bound}.
Note that the bound holds for any $\theta$ and is nontrivial provided that for all $i=1,\ldots,n$ there is a Kraus representation of the noisy channel $\mathcal{N}_{A_i}\circ\mathcal{U}_{A_i}^\theta$ satisfying $\beta_i=0$.
This is, however, a generic property of almost all realistic noise models including the erasure and depolarizing noise~\cite{Fujiwara2008, Matsumoto2010, Demkowicz2012}.

We remark that if the channels $\mathcal{N}_{A_i}\circ\mathcal{U}_{A_i}^\theta$ are all the same (which is a typical assumption in quantum metrology), then we replace the sum in Eq.~(\ref{eq_m_uses}) by a multiplicative factor $n$.
Thus, the QFI, which in principle might scale quadratically in the number of physical subsystems $n$ (the so-called Heisenberg scaling), is forced in the presence of noise to scale at most linearly.
Also, while the task of minimizing the operator norm over different equivalent Kraus representations in Eq.~(\ref{eq_m_uses}) might look challenging at first, it is only a minimization over Kraus representations $\{K_{i,k}^\theta\}_k$ of the channels $\mathcal{N}_{A_i}\circ\mathcal{U}_{A_i}^\theta$ acting independently on \emph{individual} physical subsystems $A_i$.
Such a minimization for many noise models, including the erasure and depolarizing noise, can either be done analytically or via semi-definite programs providing the optimal solution~\cite{Escher2011, Demkowicz2012, Kolodynski2013, Demkowicz2014, Demkowicz2017, Zhou2018, Zhou2020}.

\sectionprl{Application of metrological bounds to limit the quality of the QECC}

Since the metrological bound in Eq.~(\ref{eq_m_uses}) is valid for any (adaptive) metrological strategy, it is also applicable to the setting, where the initial state $\rho_A = \mathcal{E}_{L\rightarrow A}(\rho_L)$ corresponds to some encoded logical state $\rho_L$ and the control is the (adaptive) recovery operation $\mathcal{R}_{A\rightarrow L}$ followed by the encoding map $\mathcal{E}_{L\rightarrow A}$; see Fig.~\ref{fig_schemes}(b).
On the logical level, we start with an arbitrary logical state $\rho_L$, and apply $m$ times the following channel
\begin{eqnarray}
\mathcal{I}^\theta_L
&=& \mathcal{R}_{A\rightarrow L} \circ \mathcal{N}^\theta_{A} \circ \mathcal{E}_{L\rightarrow A} \\
&=& \mathcal{R}_{A\rightarrow L} \circ \bigotimes_{i=1}^n \mathcal{N}_{A_i} \circ \mathcal{E}_{L\rightarrow A} \circ \mathcal{U}^\theta_L =  \mathcal{I}_L \circ \mathcal{U}^\theta_L,
\end{eqnarray}
where we use the structure of the noisy channel $\mathcal{N}^\theta_{A}$ in Eq.~(\ref{eq_noise}) and covariance of the code $\mathcal{E}_{L\rightarrow A}$ in Eq.~(\ref{eq_covariance}) ending up with the scheme depicted in Fig.~\ref{fig_schemes}(c).

The following lemma is crucial, as it relates the maximal achievable QFI of the output of a channel with the Bures distance of the channel from the ideal unitary encoding channel.
\begin{lemma}
\label{lemma}
Let $\{ e^{-i\theta T} \}_\theta$ be a finite-dimensional representation of the unitary group $U(1)$ with a Hermitian generator $T$ and $\mathcal{U}^\theta$ be a channel corresponding to a unitary $e^{-i\theta T}$.
Then, for any $\theta$-dependent channel $\mathcal{C}^\theta$
\begin{equation}
\max_{\rho,\theta} F[\mathcal{C}^\theta (\rho)]
\geq \min_\theta \left[1- 8 \d{\mathcal{C}^\theta, \mathcal{U}^\theta}^2\right](\Delta T)^2 ,
\end{equation}
where $\Delta T$ is the difference between the maximal and minimal eigenvalues of $T$.
\end{lemma}
We present a proof of a stronger version of Lemma~\ref{lemma} in Appendix~\ref{sec_lemma}.
We also stress that in Lemma~\ref{lemma} we do not assume any covariance of the channel $\mathcal{C}^\theta$.
The following theorem is our main result.
\begin{theorem}[approximate Eastin-Knill]
\label{theorem}
Let $\mathcal{E}_{L\rightarrow A}$ be a  $U(1)$-covariant code transversally encoding logical gates $U_L^\theta = e^{-i\theta T_L}$ and the noise $\mathcal{N}_A = \bigotimes_{i=1}^n \mathcal{N}_{A_i}$ act independently on each physical subsystem $A_i$.
If $\mathcal{E}_{L\rightarrow A}$ is $\epsilon$-correctable for $\mathcal{N}_A$, then
\begin{equation}
\label{eq_thm_bound}
\epsilon \geq \frac{(\Delta T_L)^2}{3\sqrt{6}F^{\uparrow}},
\end{equation}
where $F^{\uparrow}$ is the bound on the QFI of the corresponding metrological model given by Eq.~\eqref{eq_m_uses}.
\end{theorem}

\begin{proof}
Applying Lemma~\ref{lemma} to the channels $\mathcal{C}^\theta = \left(\mathcal{I}^{\theta}_L\right)^m$ and
$\mathcal{U}^{\theta}= \mathcal{U}_L^{m \theta}$, where $T = m T_L$, we get
\begin{multline}
\label{eq_inequality}
\min_\theta \left\{1-8\dx{(\mathcal{I}^\theta_L)^m ,\mathcal{U}_L^{m\theta}}^2\right\}(m\Delta T_L)^2 \leq\\
\leq \max_{\rho_L,\theta} F[(\mathcal{I}^\theta_L)^m (\rho_L)] \leq  \max_{\rho_A,\theta,\{V_i\}_i} F(\tilde{\rho}^\theta_A) \leq mF^{\uparrow},
\end{multline}
where in the second line we follow the reasoning depicted in Fig.~\ref{fig_schemes}, i.e., that $F[(\mathcal{I}^\theta_L)^m (\rho_L)]$ can be upper-bounded by the respective metrological bound in Eq.~(\ref{eq_m_uses}).
We utilize the triangle inequality, contractive property and unitary invariance of the Bures distance to obtain
\begin{equation}
\label{eq_bures_bound}
\d{(\mathcal{I}^\theta_L)^m ,\mathcal{U}_L^{m\theta}} \leq m\d{\mathcal{I}_L,\id_L} \leq m\epsilon,
\end{equation}
where in the last step we use the fact that the code $\mathcal{E}_{L\rightarrow A}$ is $\epsilon$-correctable for the noise $\mathcal{N}_{A}$.
Note that the inequality in Eq.~\eqref{eq_bures_bound} holds for any $\theta$.
This, in turn, leads to
\begin{equation}
\label{eq:theoremm}
(m\Delta T_L)^2 \left(1- 8m^2 \epsilon^2\right) \leq mF^{\uparrow}.
\end{equation}
To get the tightest bound on $\epsilon$ we choose $m = 3 F^{\uparrow}/(2 (\Delta T_L)^2)$ and obtain the bound in Eq.~(\ref{eq_thm_bound}).
\end{proof}

We remark that for the most relevant noise models the corresponding metrological bound in Eq.~\eqref{eq_m_uses} is
\begin{equation}
\label{eq_F_form}
F^{\uparrow} =  \sum^n_{i=1}(\Delta T_{A_i})^2  g(\mathcal{N}_{A_i}),
\end{equation}
where $g (\mathcal{N}_{A_i})$ is a function that depends on the parameters of the noise.
In particular, we show in Appendix~\ref{sec:bound} that for the erasure noise (the physical subsystem $A_i$ is erased with probability $p_i$ in a heralded way) and the depolarizing noise (which replaces the state of $A_i$ by the maximally mixed state with probability $p_i$) we can choose
\begin{equation}
\label{eq_g_form}
g(\mathcal{N}_{A_i})= (1-p_i)/p_i.
\end{equation}
We can then contrast our bound in Eq.~(\ref{eq_thm_bound}) with the bound in Theorem 2 in Ref.~\cite{Faist2019}.
In both cases, for the independent erasure noise with constant loss probability, we obtain the $1/n$ scaling of the bound.
However, our Theorem~\ref{theorem}, unlike the results in Ref.~\cite{Faist2019}, is applicable to noise models beyond the erasure noise.

Analogously as in Ref.~\cite{Faist2019}, we also consider $SU(2)$-covariant codes encoding one logical qubit.
In such cases the spectra of the unitary group generators, which are equivalent to angular momentum operators, are restricted, i.e., $\Delta T_L = 1$ and $\Delta T_{A_i} \leq d_{A_i} - 1$, where $d_{A_i}$ is the dimension of the physical subsystem $A_i$.
We can then obtain the bound
\begin{eqnarray}
\epsilon \geq \frac{1}{3\sqrt{6}\sum_{i=1}^n (d_{A_i} -1)^2 g(\mathcal{N}_{A_i})}.
\end{eqnarray}
Thus, if $(d_{A_i} -1)^2 g(\mathcal{N}_{A_i})$ is upper-bounded by some constant independent of $n$, then we cannot reduce $\epsilon$ of the $\epsilon$-correctable code faster than at the $1/n$ rate.

\sectionprl{Discussion}

We remark that a version of the Eastin-Knill theorem, i.e.,
\emph{there does not exist a finite-dimensional $SU(d_L)$-covariant code with transversal logical gates and code distance $D>2$, where $d_L$ is the dimension of the logical system $L$},
follows from (approximate Eastin-Knill) Theorem~\ref{theorem}.
Namely, assume the contrapositive.
Such a code is also $U(1)$-covariant, and thus the bound in Eq.~(\ref{eq_thm_bound}) holds.
For the independent erasure noise with loss probability $p_i = p \ll 1$, the left-hand side of Eq.~(\ref{eq_thm_bound}) scales as $p^{\lfloor\frac{D+1}{2}\rfloor}$, because the code is guaranteed to correct any error of weight up to $\lfloor\frac{D-1}{2}\rfloor$.
However, using Eqs.~\eqref{eq_F_form}~and~\eqref{eq_g_form} we obtain that the right-hand side of Eq.~(\ref{eq_thm_bound}) scales as $p$.
Thus, for sufficiently small $p$ the bound in Eq.~(\ref{eq_thm_bound}) is violated, leading to a contradiction and
showing the Eastin-Knill theorem.

We emphasize that the bound in Eq.~\eqref{eq_thm_bound} in Theorem~\ref{theorem} can be derived for an
arbitrary local noise model as long as $F^{\uparrow}$ in Eq.~\eqref{eq_m_uses} is finite, which, in turn, corresponds to the condition that for every $i=1,\ldots,n$ there exists a Kraus representation $\{ K_{i,k}^\theta\}_i$ of $\mathcal{N}_{A_i}\circ\mathcal{U}_{A_i}^\theta$ satisfying $\beta_i=0$.
When this condition is not satisfied, the Heisenberg scaling of precision may in principle be attained~\cite{Kessler2014a, Dur2014, Arrad2014, Sekatski2017, Zhou2018, Zhou2020a},
as it is the case for $T_L \propto \sigma^Z$ and noise with Kraus operators in the span of $\{ \openone, \sigma^X\}$.
In such a scenario, the right-hand side of the bound in Eq.~\eqref{eq_m_uses} would scale quadratically in $m$ and our reasoning in Eq.~\eqref{eq_inequality} would fail to yield any meaningful restriction on $\epsilon$.
This is in agreement with the existence of the three-qubit repetition code, which can correct any single-qubit bit-flip error and at the same time has logical gates $U^\theta_L = e^{-i\theta\sigma_L^Z}$ implemented transversally via $U^\theta_A = e^{-i\theta(\sigma_1^Z+\sigma_2^Z+\sigma_3^Z)}$.

Finally, we believe that studies of correlated noise models in quantum metrological problems~\cite{Jeske_2014, layden2018spatial, Czajkowski2019, Chabuda2020}
may allow us to use our proof techniques beyond the setting of local noise independently affecting each physical subsystem.

\sectionprl{Acknowledgements}

The authors thank Aidan Chatwin-Davies, Philippe Faist, Felix Leditzky, Tobias J. Osborne, Amit Kumar Pal and Fernando Pastawski for valuable discussions.
A.K. acknowledges funding provided by the Simons Foundation through the ``It from Qubit'' Collaboration.
Research at Perimeter Institute is supported in part by the Government of Canada through the Department of Innovation, Science and Economic Development Canada and by the Province of Ontario through the Ministry of Colleges and Universities.
RDD acknowledges support from the National Science Center (Poland) grant No.\ 2016/22/E/ST2/00559.

\onecolumngrid
\appendix

\section{Proof of Lemma~\ref{lemma}}
\label{sec_lemma}

Below we prove a stronger version of Lemma~\ref{lemma} from the main text.
Instead of maximizing the QFI on the left-hand side of Eq.~\eqref{eq_thm_bound} over input states $\rho$, we explicitly provide a concrete state for which the bound holds.
\setcounter{lemma}{0}
\begin{lemma}[stronger version]
Let $\{ e^{-i\theta T} \}_\theta$ be a finite-dimensional representation of the unitary group $U(1)$ with a Hermitian generator $T$ and $\mathcal{U}^\theta$ be a channel corresponding to a rotation $e^{-i\theta T}$.
Let $\ket{t_+}$ and $\ket{t_-}$ be two orthogonal and normalized eigenvectors of $T$, which correspond to the maximal $t_+$ and minimal $t_-$ eigenvalues of $T$.
Then, for $\ket{\psi} = \frac{\ket{t_+} + \ket{t_-}}{\sqrt 2}$ and any $\theta$-dependent channel $\mathcal{C}^\theta$ we have
\begin{equation}
\label{eq_lemma_strong}
\max_\theta F[\mathcal{C}^\theta (\ketbra{\psi}{\psi})]
\geq \min_\theta \left[1- 8 \d{\mathcal{C}^\theta, \mathcal{U}^\theta}^2\right](t_+ - t_-)^2.
\end{equation}
\end{lemma}

\begin{proof}
We start by defining a $U(1)$-convolved channel $\con$ for the channel $\mathcal{C}^\theta$ in the following way
\begin{equation}
\con = \intd{\theta'}{\mathcal{U}^{-\theta'}\circ\mathcal{C}^{\theta+\theta'}},
\end{equation}
where $d\theta'$ corresponds to the normalized Haar measure of the group $U(1)$, i.e., $\int \t{d} \theta' = 1$ and the interval of integration is $[0,2\pi]$.
By making a substitution $\theta'' = \theta + \theta'$ we obtain
\begin{equation}
\con = \intd{\theta''}{\mathcal{U}^{\theta}\circ\mathcal{U}^{-\theta''}\circ\mathcal{C}^{\theta''}}
= \mathcal{U}^{\theta}\circ\conx,
\end{equation}
where we define a $\theta$-independent channel $\conx = \intd{\theta'}{\mathcal{U}^{-\theta'}\circ\mathcal{C}^{\theta'}}$.

We can show that the Bures distance between the $U(1)$-convolved channel $\con$ and the unitary rotation channel $\mathcal{U}^\theta$ is upper-bounded by the distance between $\mathcal{C}^\theta$ and $\mathcal{U}^\theta$ maximized over $\theta$, i.e.,
\begin{equation}
\label{eq_dist_con}
\d{\con,\mathcal{U}^\theta} \leq \max_\theta \d{\mathcal{C}^\theta,\mathcal{U}^\theta}.
\end{equation}
We infer the above inequality from the following sequence of equalities and inequalities
\begin{eqnarray}
\left[1-\d{\con,\mathcal{U}^\theta}^2\right]^2
&=& [\mathcal{F}(\con, \mathcal{U}^\theta)]^2 = [\mathcal{F}(\conx,\id)]^2
= \min_{\ket{\Phi}} \bra{\Phi} \left[\conx(\ketbra{\Phi}{\Phi})\right]\ket{\Phi}\\
&=& \min_{\ket{\Phi}} \bra{\Phi} \left[\intd{\theta'}
{\mathcal{U}^{-\theta'}\circ\mathcal{C}^{\theta'}(\ketbra{\Phi}{\Phi})}\right] \ket{\Phi}
\geq \intd{\theta'}{\min_{\ket{\Phi}} \bra{\Phi}
\left[\mathcal{U}^{-\theta'}\circ\mathcal{C}^{\theta'} (\ketbra{\Phi}{\Phi})\right] \ket{\Phi}}\\
&=& \intd{\theta'}{\left[\mathcal{F}(\mathcal{U}^{-\theta'}\circ\mathcal{C}^{\theta'},\id)\right]^2}
\geq \min_\theta \left[\mathcal{F}(\mathcal{C}^{\theta},\mathcal{U}^{\theta})\right]^2
= \min_\theta \left[1-\d{\mathcal{C}^\theta,\mathcal{U}^\theta}^2\right]^2.
\end{eqnarray}

We can upper-bound the QFI of the output state $\con(\rho)$ of the $U(1)$-convolved channel $\con$ by the QFI of $\mathcal{C}^\theta(\rho)$ maximized over $\theta$, where $\rho$ is an arbitrary input state.
Namely,
\begin{eqnarray}
\label{eq_fisher_con}
F[\con(\rho)] &=& F\left[\intd{\theta'}{\mathcal{U}^{-\theta'}\circ\mathcal{C}^{\theta+\theta'}(\rho)}\right]
\leq \intd{\theta'}{F\left[\mathcal{U}^{-\theta'}\circ\mathcal{C}^{\theta+\theta'}(\rho)\right]} =
\intd{\theta'}{F\left[\mathcal{C}^{\theta+\theta'}(\rho)\right]} \leq \max_\theta F\left[\mathcal{C}^{\theta}(\rho)\right],
\end{eqnarray}
where in the first inequality we use convexity of the QFI, and then in the following equality we use invariance of the QFI under the $\theta$-independent unitary rotation channel $\mathcal{U}^{-\theta'}$.

Now, we are going to derive the following lower bound on the QFI of $\con(\ketbra{\psi}{\psi})$
\begin{equation}
\label{eq_fisher_psi}
F[\con(\ketbra{\psi}{\psi})] \geq \left[1 - 8\d{\con,\mathcal{U}^\theta}^2\right] (t_+ - t_-)^2,
\end{equation}
which combined with Eqs.~\eqref{eq_fisher_con}~and~\eqref{eq_dist_con} leads to the inequality in our lemma.
Let $\ket{t_-}, \ket{t_1}, \ldots, \ket{t_\tau},\ket{t_+}$ denote all mutually orthogonal
and normalized eigenvectors of $T$, whose corresponding eigenvalues form a non-decreasing sequence,
 i.e., $t_- \leq t_1 \leq \ldots \leq t_\tau \leq t_+$.
Let $\mathcal{H}_\pm = \t{span}(\ket{t_+},\ket{t_-})$ be the space spanned by $\ket{t_+}$ and $\ket{t_-}$ and $\mathcal{H}^\perp_\pm = \t{span}(\ket{t_1},\ldots,\ket{t_\tau})$ be the space orthogonal to $\mathcal{H}_\pm$.
We define a decohering channel $\mathcal{D}$ as follows
\begin{equation}
\mathcal{D}(\cdot) = \Pi_\pm \cdot \Pi_\pm + \Pi^\perp_\pm \cdot \Pi^\perp_\pm,
\end{equation}
where $\Pi_\pm$ and $\Pi^\perp_\pm$ denote the projectors onto $\mathcal{H}_\pm$ and  $\mathcal{H}^\perp_\pm$, respectively.
Note that $\mathcal{D}$ does not depend on the parameter $\theta$ and
$\mathcal{D}\circ \mathcal{U}^\theta = \mathcal{U}^\theta \circ \mathcal{D}$.
Lastly, we choose an orthonormal basis $\mathcal{B} = \{\ket{\psi}, \ket{\psi^\perp}, \ket{t_1}\ldots, \ket{t_\tau}\}$, where $\ket{\psi^\perp} = \frac{\ket{t_+}-\ket{t_-}}{\sqrt{2}}$.

Let $f$ denote the fidelity between $\conx(\ketbra{\psi}{\psi})$ and $\ket{\psi}$, i.e.,
$f = f[\conx(\ketbra{\psi}{\psi}),\ketbra{\psi}{\psi}]$.
Then, we have
\begin{eqnarray}
\label{eq_fid_con}
f^2 &=& \{f[\mathcal{U}^\theta\circ\conx(\ketbra{\psi}{\psi}),\mathcal{U}^\theta(\ketbra{\psi}{\psi})] \}^2
\geq [\mathcal{F}(\con,\mathcal{U}^\theta)]^2
= \left[1-\d{\con,\mathcal{U}^\theta}^2\right]^2 \geq 1-2\d{\con,\mathcal{U}^\theta}^2.
\end{eqnarray}
Also, we can express $\conx(\ketbra{\psi}{\psi})$ in a generic block form in the basis $\mathcal{B}$ as
\begin{equation}
\conx(\ketbra{\psi}{\psi}) = \mat{c|c}{
(f^2+q-f^2 q)\rho_{\pm} &  *\\
\hline
* & (1-q)(1-f^2) \rho^{\perp}_{\pm}},\quad
\rho_{\pm} = \frac{1}{f^2+q-f^2q}\mat{cc}{ f^2 & c \\ c^* & q-f^2 q},
\end{equation}
where $q\in[0,1]$, $c$ is some complex number and the density matrices $\rho_{\pm}$ and $\rho^\perp_{\pm}$ are supported on $\mathcal{H}_\pm$ and  $\mathcal{H}^\perp_\pm$, respectively.
Note that $\rho_{\pm}$ and $\rho^\perp_{\pm}$ do not depend on the parameter $\theta$.
Let us write the generator $T$ in the basis $\mathcal{B}$ as
\begin{equation}
T = T_\pm \oplus T^\perp_\pm,\quad
T_{\pm} = \frac{1}{2}\mat{cc}{t_+ + t_- & t_+ - t_- \\ t_+ - t_- & t_+ + t_-},\quad
T_\pm^\perp = \t{diag}(t_1,\ldots,t_\tau),
\end{equation}
where operators $T_\pm$ and  $T^\perp_\pm$ act on $\mathcal{H}_\pm$ and  $\mathcal{H}^\perp_\pm$.
Then, we have
\begin{eqnarray}
F[\con(\ketbra{\psi}{\psi})] &=& F[\mathcal{U}^\theta\circ\conx(\ketbra{\psi}{\psi})]
\geq F[\mathcal{D}\circ \mathcal{U}^\theta\circ\conx(\ketbra{\psi}{\psi})]
= F[\mathcal{U}^\theta\circ\mathcal{D}\circ \conx(\ketbra{\psi}{\psi})] \\
&=& F\left[(f^2+q-f^2 q) e^{-i\theta T_\pm} \rho_{\pm} e^{i\theta T_\pm}
\oplus (1-q)(1-f^2) e^{-i\theta T^\perp_\pm} \rho^{\perp}_{\pm} e^{i\theta T^\perp_\pm}\right]\\
&=& (f^2+q-f^2 q) F\left(e^{-i\theta T_\pm} \rho_{\pm} e^{i\theta T_\pm}\right) +
(1-q)(1-f^2) F\left(e^{-i\theta T^\perp_\pm} \rho^{\perp}_{\pm} e^{i\theta T^\perp_\pm}\right)\\
&\geq& (f^2+q-f^2 q) F\left(e^{-i\theta T_\pm} \rho_{\pm} e^{i\theta T_\pm}\right),
\end{eqnarray}
where we use the facts that: (i) the QFI does not increase under the parameter-independent channel $\mathcal{D}$
and (ii) for any state with a block-diagonal structure, which is invariant under the parameter change, the QFI is a weighted sum of the QFI corresponding to the respective blocks.
Let $\ket{i}$ be an eigenvector of $\rho_\pm$ corresponding to an eigenvalue $\lambda_i$ for $i=1,2$.
Then, using a general formula for the QFI in a unitary parameter-encoding model we obtain
\begin{equation}
F(e^{-i \theta T_\pm } \rho_\pm e^{i\theta T_\pm})
= \sum_{i,j=1}^2 \frac{2 \lvert\bra{i} T_\pm \ket{j}\rvert^2(\lambda_i - \lambda_j)^2}{\lambda_i + \lambda_j}
= \left(\frac{f^2-q+f^2 q}{f^2+q-f^2 q}\right)^2 (t_+ - t_-)^2,
\end{equation}
which leads to the following lower-bound
\begin{equation}
F[\con(\ketbra{\psi}{\psi})] \geq \frac{(f^2-q+f^2 q)^2}{f^2+q-f^2 q} (t_+ - t_-)^2.
\end{equation}
The above inequality has to hold for any $q\in[0,1]$.
In particular, by setting $q=1$ (which also happens to maximize the right-hand side of the above inequality for $f\geq \frac{1}{2}$) and using Eq.~\eqref{eq_fid_con} we finally arrive at
\begin{equation}
F[\con(\ketbra{\psi}{\psi})] \geq (2f^2-1)^2 (t_+ - t_-)^2 \geq (4f^2-3) (t_+ - t_-)^2
\geq \left[1-8\d{\mathcal{C}^\theta,\mathcal{U}^\theta}^2\right] (t_+ - t_-)^2,
\end{equation}
which gives Eq.~\eqref{eq_fisher_psi} and finishes the proof of our lemma.
\end{proof}

\section{Derivation of fundamental metrological bounds}
\label{sec:bound}

For completeness, we provide a derivation of fundamental metrological bounds along the lines developed in~\cite{Escher2011,Demkowicz2012, Kolodynski2013,Demkowicz2014} with a straightforward generalization, which allows the channels acting on different physical subsystems to be different.
Similarly as in~\cite{Demkowicz2014}, we consider the most general adaptive strategy when deriving the bound.

Consider an action of a quantum channel $\mathcal{N}^\theta$ on some state $\ket{\psi}$, which can be written using its Kraus representation
\begin{equation}
\rho^\theta = \mathcal{N}^{\theta}(\ketbra{\psi}{\psi})  = \sum_k K_k^\theta \ketbra{\psi}{\psi} K_k^{\theta\dagger}.
\end{equation}
Let $W^{\theta}$ be a unitary dilation of $\mathcal{N}^\theta$ on an extended Hilbert space $\mathcal{H}\otimes \mathcal{H}_{E}$,
where $\mathcal{H}_E$ is an ancillary space, such that
\begin{equation}
\mathcal{N}^{\theta}(\ketbra{\psi}{\psi}) = \t{Tr}_E\left( W^{\theta} \ketbra{\psi}{\psi} \otimes \ketbra{0}{0}_E  W^{\theta \dagger}  \right),
\end{equation}
where $\ket{0}$ is some arbitrary state of the ancillary system.
The Kraus operators are related to a particular unitary dilation $W^\theta$ via $K_k^{\theta} = {_E}\bra{k} W^{\theta} \ket{0}_E$.
We can write a purification of the output state as
\begin{equation}
\ket{\Psi^{\theta}} = W^{\theta} \ket{\psi} \otimes \ket{0} = \sum_k K_k^{\theta} \otimes \openone \ket{\psi} \otimes \ket{0}.
\end{equation}
The QFI of $\rho^{\theta}$ is upper-bounded by the QFI of its purification.
Using the formula for the QFI of a pure state we have
\begin{equation}
F(\ket{\Psi(\theta}) = 4 \left(\braket{\dot{\Psi}^{\theta}}{\dot{\Psi}^{\theta}}  - \left|\braket{\dot{\Psi}^{\theta}}{\Psi^{\theta}}\right|^2 \right) =
4 \left(\bra{\psi}\sum_k \dot{K}_k^{\theta \dagger}  \dot{K}_k^{\theta} \ket{\psi} -
\left| \sum_k \bra{\psi}\dot{K}^{\theta \dagger}_k K_k^{\theta} \ket{\psi}\right|^2  \right)
\leq 4 \bra{\psi}\sum_k \dot{K}^{\theta \dagger}\dot{K}_{k}^{\theta} \ket{\psi},
\end{equation}
where $|\dot{\Psi}^\theta\rangle$ and $\dot{K}_k^{\theta}$ denote the derivatives of $\ket{\Psi^\theta}$ and $K_k^{\theta}$ with respect to $\theta$.
Since the above inequality is valid for any purification of the output state, we can obtain the tightest bound by minimizing the right-hand side over different purifications, which, in turn, is equivalent to minimizing over Kraus representations $\{K_k^\theta\}_k$ of $\mathcal{N}^{\theta}$. Thus,
\begin{equation}
F(\rho^{\theta}) \leq 4 \min_{\{K_k^{\theta}\}_k}\bra{\psi}\sum_k \dot{K}_k^{\theta \dagger}  \dot{K}_k^{\theta} \ket{\psi}.
\end{equation}
Finally, we may bound the QFI optimized over input states as follows
\begin{equation}
\label{eq:qfibound}
\max_{\ket{\psi}} F(\rho^{\theta}) \leq 4 \max_{\ket{\psi}} \min_{\{K_k^{\theta}\}_k} \bra{\psi}\sum_k \dot{K}_k^{\theta \dagger}
\dot{K}_k^{\theta} \ket{\psi}
\leq 4 \min_{\{K_k^{\theta} \}_k} \max_{\ket{\psi}} \bra{\psi}\sum_k \dot{K}_k^{\theta \dagger}\dot{K}_k^{\theta} \ket{\psi}
= 4 \min_{\{K_k^{\theta}\}_k} \left\|  \sum_k \dot{K}_k^{\theta \dagger}  \dot{K}_k^{\theta} \right\|,
\end{equation}
where $\| \cdot \|$ is the operator norm and we use the max-min inequality.
As a result, we have an upper bound on the achievable QFI of the output state as a function of the Kraus operators defining the quantum channel itself.

\begin{figure}[t]
\centering
\includegraphics[width=0.55\columnwidth]{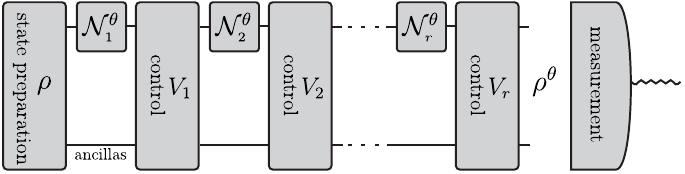}
\caption{
The most general form of an adaptive metrological scheme, which uses $\r$  parameter-encoding channels $\mathcal{N}_1^\theta,\ldots,\mathcal{N}_r^\theta$.
The control unitaries $\{V_i\}_i$ allow to entangle the probe with an arbitrary large ancillary system, which
together with a collective measurement of the output state $\rho^\theta$ guarantees the full generality of the setup.}
\label{fig:adaptive}
\end{figure}

The power of the above formula only becomes evident when one views the channel $\mathcal{N}^{\theta}$ as composed of $r$ independent channels $\mathcal{N}_1^{\theta},\ldots,\mathcal{N}_r^{\theta}$; see Fig.~\ref{fig:adaptive}.
Here we consider the most general adaptive strategy, where each individual channel $\mathcal{N}_i^\theta$ is followed by a unitary $V_i$ that potentially entangles the probe with an arbitrary large ancillary system.
At the end, all the systems are measured using a collective measurement.
In particular, if we appropriately choose the unitaries $V_i$ to be swap operators, then this scheme can encompass  a parallel scheme, where $r$ different subsystems are initially prepared in some arbitrary
(possibly entangled) state and $r$ channels $\mathcal{N}_1,\ldots,\mathcal{N}_r$ act simultaneously on the corresponding subsystems $1,\ldots,r$.

The action of the the channel $\mathcal{N}^{\theta}$ can be written as
\begin{equation}
\rho^{\theta}=\mathcal{N}^{\theta}(\rho) = \sum_{\boldsymbol{k}}K_{\boldsymbol{k}}^{\theta} \rho K_{\boldsymbol{k}}^{\theta \dagger},
\quad K_{\boldsymbol{k}}^{\theta} = K_{\r,k_{\r}}^{\theta}V_{\r} K_{\r-1,k_{\r-1}}^{\theta}V_{\r-1}
\ldots K_{1,k_1}^{\theta} V_1\textrm{ for $\boldsymbol{k}=\{k_1,\dots,k_{\r}\}$},
\end{equation}
where we use a boldface index $\boldsymbol{k}$ to combine all the indices of Kraus operators of individual channels $\mathcal{N}_i$.
Let us for a moment fix the Kraus representation of each channel $\mathcal{N}_i$.
For brevity of the notation, we write $\mathrm{K}_{i,k} := K_{i,k}^{\theta} V_i$, where we suppress the explicit dependence on $\theta$.
According to Eq.~\eqref{eq:qfibound}, the upper bound on the QFI of the output state is
 \begin{equation}
 F(\rho^{\theta}) \leq 4  \left\| \sum_{\boldsymbol{k}} \dot{K}_{\boldsymbol{k}}^{\theta \dagger} \dot{K}_{\boldsymbol{k}}^{\theta}\right\| = 4  \left\|\sum_{\boldsymbol{k}} \sum_{i,j=1}^{\r} \mathrm{K}_{1,k_1}^\dagger  \dots\dot{\mathrm{K}}_{i,k_i}^\dagger \dots \mathrm{K}_{\r,k_{\r}}^\dagger \mathrm{K}_{\r,k_{\r}}  \dots \dot{\mathrm{K}}_{j,k_j} \dots  \mathrm{K}_{1,k_1} \right\|.
 \end{equation}
 Note the following property of the operator norm
 \begin{equation}
 \label{eq:normineq}
 \left\|{\sum_k L_{k}^{\dagger} A L_k}\right\| \leq \| A \|\,  \left\| \sum_k L_k^{\dagger} L_k \right\|
 \end{equation}
 valid for any operators $A$ and $L_k$. Making use of the above  property together with the triangle inequality and
  the trace preserving condition $\sum_k \mathrm{K}_{k}^{\dagger} \mathrm{K}_k = \openone$ we get
  \begin{equation}
  \label{eq:fishbound1}
 F(\rho^{\theta}) \leq 4 \sum_i \left\| \sum_{k_i}\dot{\mathrm{K}}_{i,k_i}^{\dagger} \dot{\mathrm{K}}_{i,k_i}^{ \dagger} \right\|
  +\sum_{i<j} \left\| \sum_{k_i,\dots,k_j}\dot{\mathrm{K}}^{k_i \dagger}_{i} \dots {\mathrm{K}}^{k_j \dagger}_{j} \dot{\mathrm{K}}^{k_j}_j
  \dots {\mathrm{K}}^{k_i}_{i}  + h.c.\right\|.
\end{equation}
When analysing  the second summation term, note that the trace preserving condition implies that an operator $iA = \sum_k \dot{\mathrm{K}}^{k \dagger} \mathrm{K}^k$ is anti-hermitian.
Then, we have
\begin{equation}
\left\|\sum_k \dot{\mathrm{K}}^{k \dagger}  i A \mathrm{K}^k + h.c.\right\| 
= \left\| \sum_k (\dot{\mathrm{K}}^k + i \mathrm{K}^k)^\dagger A (\dot{\mathrm{K}}^k + i \mathrm{K}^k) -
\dot{\mathrm{K}}^{k \dagger}  A \dot{\mathrm{K}}^k - {\mathrm{K}}^{k \dagger}  A \mathrm{K}^k \right\|.
\end{equation}
Using the triangle inequality together with Eq.~\eqref{eq:normineq} we get
\begin{equation}
\left\|\sum_k \dot{\mathrm{K}}^{k \dagger}  i A \mathrm{K}^k + h.c.\right\| \leq 2
\| A\|\left(\left\|\sum_k \dot{\mathrm{K}}^{k \dagger} \dot{\mathrm{K}}^k\right\| + \left\| \sum_k \dot{\mathrm{K}}^{k \dagger}
 \mathrm{K}^{k \dagger} \right\| +1\right).
\end{equation}
This leads us to the final bound on the QFI, where we include a minimization over the Kraus representations
\begin{equation}
\label{eq:unversalbound}
F(\rho^{\theta}) \leq 4  \left(\min_{\{K_{i,k}^{\theta}\}}
\sum_i \|\alpha_i\| + \sum_{i \neq j} \|\beta_i\| (\|\alpha_j\| + \| \beta_j\| + 1)\right),
\quad \alpha_i = \sum_k \dot{K}_{i,k}^{\theta \dagger}\dot{{K}}_{i,k}^{\theta},
\quad \beta_i= \sum_{k} \dot{{K}}_{i,k}^{\theta \dagger}  K_{i,k}^{\theta}.
\end{equation}
Note that we can substitute $\mathrm{K}_{i,k}$ with $K_{i,k}^{\theta}$ as this replacement amounts to multiplying Kraus operators by unitary operations which does not affect the operator norms in Eq.~\eqref{eq:unversalbound}.

For typical models of decoherence, such as the erasure or depolarizing noise, one can find Kraus representations $\{ K^\theta_{i,k}\}_k$ satisfying $\beta_i=0$, which leads to powerful upper bounds that scale linearly in the number of subsystems, i.e.,
\begin{equation}
\label{eq:unversalboundbeta0}
F(\rho^{\theta}) \leq 4  \sum_i \min_{\{K_{i,k}^{\theta}\}_k, \beta_i=0} \|\alpha_i\| \leq 4 \r \max_i \min_{\{K_{i,k}^{\theta}\}_k, \beta_i=0} \|\alpha_i\|.
\end{equation}
This, in turn, indicates the impossibility of achieving the Heisenberg scaling of precision
\cite{Fujiwara2008, Escher2011,Demkowicz2012, Kolodynski2013,Demkowicz2014}.
Note that we can also arrive at the same conclusion as long as $\beta_i = 0$ for all except for some constant number of channels.

In order to obtain an explicit bound one needs to perform a minimization over Kraus representations $\{K_{i,k}^\theta\}_k$ for each individual channel $\mathcal{N}_i$.
This can  be easily done numerically using a semi-definite program as described in \cite{Demkowicz2012, Kolodynski2013, Demkowicz2014, Demkowicz2017, Zhou2018, Zhou2020, Zhou2020}.
Here, we will not discuss this optimization procedure; rather, we simply provide the optimal Kraus representations
for the erasure and depolarizing noise.
We remark that since the bound is given in terms of a minimization over Kraus representations, any Kraus
representation provides a valid bound.

Note that in our work we consider the adaptive scheme as depicted in Fig.~\ref{fig_schemes}(a) with the total number of channels $r = n m$, where $n$ represents the number of physical subsystems on which respective parameter-encoding channels act in parallel, and $m$ is the number of adaptive steps.
As such, Fig.~\ref{fig_schemes}(a) is a special case of the more general Fig.~\ref{fig:adaptive}.

\subsection{Erasure noise}

In order to describe the erasure channel $\mathcal{N}_e$ acting on a $d$-dimensional subsystem it is convenient to introduce an additional level of the subsystem labeled by $\ket{d+1}$.
This way we can formally work with the same number of subsystems and the fact that a given subsystem is lost is simply indicated by its internal state being $\ket{d+1}$.
We therefore work using the basis $\{\ket{1},\dots,\ket{d}, \ket{d+1}\}$.
The canonical Kraus operators for the erasure noise are
\begin{equation}
K_0 = \sqrt{1-p} \sum_{i=1}^d \ketbra{i}{i},\quad
K_{d+1} =\ketbra{d+1}{d+1},\quad
K_k = \sqrt{p} \ketbra{d+1}{k}\textrm{ for $k=1,\dots, d$},
\end{equation}
where $p$ is the loss probability, $K_0$ represents the event of no loss, $K_{k}$ represents the event when the subsystem is in a state $\ket{k}$, which is lost, and $K_{d+1}$ represents simply the fact that a lost subsystem remains lost.

When the unitary parameter encoding with a generator $T = \t{diag}(t_1,\dots,t_d)$ is additionally considered we get
\begin{equation}
K_k^{\theta} = K_k e^{-i T \theta}.
\end{equation}
Note that the generator $T$ should be formally understood as $T \oplus 0$, where $0$ stands for the lack of phase encoding on the lost subsystem.
Without loss of generality we assume that the generator $T$ is shifted in a way that the maximal $t_+$ and minimal $t_-$ eigenvalues have the same absolute values, i.e., $t_+ = -t_- = \Delta T/2$.

The above Kraus representation is not helpful with obtaining strong metrological bounds as $\beta \neq 0$
and the bound would scale quadratically in the number of subsystems.
However, we can consider an equivalent Kraus representation
\begin{equation}
\tilde{K}_0^{\theta} = K_0^{\theta}, \quad  \tilde{K}_k^{\theta} = e^{i c_k \theta} K_k^{\theta}, \quad \tilde{K}_{d+1}^{\theta} = K_{d+1}^{\theta}.
\end{equation}
When evaluated at $\theta = 0$, this representation results in operators $\alpha$ and $\beta$ of the form
\begin{equation}
\beta = i  \sum_{k=1}^{d+1} (t_k-p c_k) \ketbra{k}{k}, \quad \alpha = \sum_{k=1}^{d+1} [(1-p) t_k^2+p(c_k-t_k)^2]\ketbra{k}{k}.
\end{equation}
In order to have $\beta=0$ we need to choose $c_k = t_k/p$.
Then, the diagonal elements of $\alpha$ are
$t_k^2 (1-p)/p$ and the operator norm of $\alpha$ corresponds to their largest absolute value, which is
$t_+^2 (1-p)/p = (\Delta T)^2 (1-p)/(4p)$.
Thus, when we consider $r$ erasure channels, each with the corresponding loss probability $p_i$,
the QFI of the output state of an arbitrary adaptive strategy that involves using these $r$ channels
is upper-bounded by
\begin{equation}
\label{eq:qfierasure}
F(\rho^{\theta}) \leq F^{\uparrow}_e =  \sum_{i=1}^{\r} (\Delta T_i)^2 \frac{1-p_i}{p_i},
\end{equation}
where $\Delta T_i$ is the spectrum spread of the generator $T_i$ acting on the subsystem $i$.

\subsection{Depolarizing channel}

The simplest way to derive a bound for the depolarizing channel $\mathcal{N}_{d}$ is to note that it can be regarded as a composition of the erasure channel $\mathcal{N}_{e}$ described in the previous subsection with a channel that takes the state $\ket{d+1}$ and returns the maximally mixed state on the $d$-dimensional Hilbert space spanned by $\ket{1},\dots,\ket{d}$ while on other states it acts as the identity channel.
Since the QFI does not increase under the action of any parameter-independent channels,
any bound derived for the erasure channel with a given loss probability $p$ is valid for the depolarization channel, which replaces the state with the maximally mixed state with probability $p$.
However, such a bound is not tight.
To obtain a better bound one needs to perform a minimization over Kraus representations of the actual channel.
Below we provide an analytical form of the optimal bound for $d=2$.
This tighter bound may be used in Eq.~\eqref{eq_thm_bound} in the main text provided physical subsystems $A_i$ correspond to qubits.

Kraus operators of the qubit isotropic depolarizing model are
\begin{equation}
K_{0}=\sqrt{1- \frac{3p}{4}} \openone, \quad K_k=\sqrt{\frac{p}{4}}\,\sigma^{k} \textrm{ for $k=1,2,3$},
\end{equation}
where $\sigma^k$ are Pauli matrices, and $1-p$ is the effective shrinking factor of the qubit's Bloch vector.
We consider the unitary encoding operation, which rotates the qubit around the $z$ axis, i.e., $U^{\theta} = \exp(- i T \theta )$ with the generator $T = \sigma^Z \Delta T/2$, where $\Delta T$ represent the spectrum spread of the generator. The effective Kraus operators are
$K_k^{\theta} = K_k e^{-i T \theta}$. This representation again does not yield $\beta=0$, and
we need to find another representation to obtain a tighter bound. The optimal Kraus representation has the following structure
\begin{equation}
\mat{c}{\tilde{K}_0^{\theta}\\
\tilde{K}_1^{\theta}\\
\tilde{K}_2^{\theta}\\
\tilde{K}_3^{\theta}
} = \mat{cccc}{\cos (a\theta)& 0 & 0 & i\sin(a\theta) \\
0& \cos(b\theta) & -  \sin (b\theta) & 0 \\
0& \sin (b \theta) & \cos(b\theta) & 0 \\
i\sin(a\theta)& 0 & 0 & \cos(a\theta)}
\mat{c}{{K}_0^{\theta}\\
{K}_1^{\theta}\\
{K}_2^{\theta}\\
{K}_3^{\theta}}.
\end{equation}
Then, the operators $\alpha$ and $\beta$ are
\begin{equation}
\alpha =\frac{1}{4}\left((\Delta T)^2  +\Delta T\left(2 b p - 2a \sqrt{p(4-3p)}\right)+ 2 b^2p - 2a^2(p-2)\right) \openone, \quad \beta = \frac{i}{2}\left( \Delta T + b p -a \sqrt{p(4 - 3p)}\right) \sigma^Z.
\end{equation}
In order to have $\beta=0$ we set $a= (b p +\Delta T)/\sqrt{p(4-3p)}$.
Then, the minimal value of the norm of $\alpha$ is admitted for $b=(2-p)\Delta T/(2p)(2p-3)$
and equals $\|  \alpha \| = (\Delta T)^2(1-p)^2/(2p(2p-3))$.
Thus, given $\r$ different subsystems the fundamental bound on estimating $\theta$ in the presence of the depolarizing noise is
\begin{equation}
F[\rho^{\theta}] \leq F^{\uparrow}_d = \sum_{i=1}^{\r} (\Delta T_i)^2 \frac{2 (1-p_i)^2}{p_i(2p_i -3)},
\end{equation}
where $p_i$ is the strength of the depolarizing noise on the subsystem $i$.
If $p_i < 1/2$, then the right-hand side is indeed smaller than the right-hand side of Eq.~\eqref{eq:qfierasure} and hence the depolarizing bound is tighter than the erasure bound.

\twocolumngrid
\bibliography{EK_metro}
\end{document}